\documentclass[%
 reprint,
 amsmath,amssymb,
 aps,
 prstab,
]{revtex4-1}

\usepackage{graphicx}
\usepackage{dcolumn}
\usepackage{bm}

\usepackage{amsthm}

\newcommand{\dw}{\delta_{\text{worst}}}
\newcommand{\da}{\delta_{\text{avg}}}
\newcommand{\dista}{d_{\text{avg}}}
\newcommand{\diam}{D}
\renewcommand{\O}{\mathcal{O}}
\newcommand{\datasetsize}{93}
\newcommand{\datasetbio}{19}
\newcommand{\datasetsoc}{32}
\newcommand{\datasettec}{42}
\renewcommand{\P}{\mathbb{P}}

\newtheorem{proposition}{Proposition}

\begin{document}
\title{Hyperbolicity Measures ``Democracy'' in Real-World Networks}
\thanks{\small{Correspondence and requests for materials should be addressed to M.B. (\emph{michele.borassi@imtlucca.it})}}%

\author{Michele Borassi}
\affiliation{IMT Institute for Advanced Studies, Piazza San Francesco 19, 55100 Lucca, Italy}
\author{Alessandro Chessa}
\affiliation{IMT Institute for Advanced Studies, Piazza San Francesco 19, 55100 Lucca, Italy}
\affiliation{Linkalab, Complex Systems Computational Laboratory, 09129 Cagliari, Italy}

\author{Guido Caldarelli}
\affiliation{IMT Institute for Advanced Studies, Piazza San Francesco 19, 55100 Lucca, Italy}
\affiliation{Istituto dei Sistemi Complessi (ISC), via dei Taurini 19, 00185 Roma, Italy}
\affiliation{London Institute for Mathematical Sciences, 35a South St. Mayfair W1K 2XF London UK}

\vspace{0.5cm}



\begin{abstract}
We analyze the hyperbolicity of real-world networks, a geometric quantity that measures if a space is negatively curved. In our interpretation, a network with small hyperbolicity  is ``aristocratic'', because it contains a small set of vertices involved in many shortest paths, so that few elements ``connect'' the systems, while a network with large hyperbolicity has a more ``democratic'' structure with a larger number of crucial elements.

We prove mathematically the soundness of this interpretation, and we derive its consequences by analyzing a large dataset of real-world networks. We confirm and improve previous results on hyperbolicity, and we analyze them in the light of our interpretation.

Moreover, we study (for the first time in our knowledge) the hyperbolicity of the neighborhood of a given vertex. 
This allows to define an ``influence area'' for the vertices in the graph. We show that the influence area of the highest degree vertex is small in what we define ``local'' networks, like most social or peer-to-peer networks. On the other hand, if the network is built in order to reach a ``global'' goal, as in metabolic networks or autonomous system networks, the influence area is much larger, and it can contain up to half the vertices in the graph. 
In conclusion, our newly introduced approach allows to distinguish the topology and the structure of various complex networks.

\end{abstract}
\maketitle

\section{Introduction}
The most basic way to describe a graph is to consider its metric quantities 
as for instance the diameter \cite{bollobas2004diameter}, 
the degrees, the distances \cite{bollobas1979graph}. 
In case no other information is available, a good choice is to consider randomly drawn edges \cite{erdos1959random,erdos1960evolution} and determine the expected values and distribution of 
those properties in such random graphs. More recently, the computer revolution and the pervasive presence of Internet and WWW, created a whole series of
complex networks in technological systems whose properties can be directly measured from data \cite{albert2002statistical}. All the real networks show particular structures of edges, making them definitely different from random graphs. Driven by such evidence, researchers  recognized analogous structures in a variety of other cases, ranging from biology to economics and finance \cite{caldarelli2007scale}.
All these structures show lack of characteristic scale in the statistical distribution of the degree and small world effect, making therefore important to understand the basic principles at the basis of their formation \cite{barabasi1999emergence,caldarelli2002scale,krioukov2010hyperbolic}. 

To bring order in this huge set of systems, it would be extremely useful if we could classify the various networks by means of some specific quantity differing from case to case. 
In this quest of distinguishing universal from particular behavior we decide to  consider the connection with the ``curvature'' of the graph. 
Embedding spaces can have  negative curvature (hyperbolic spaces), zero value of curvature (Euclidean spaces), or positive curvature (spherical spaces).
On the bases of the hyperbolicity measure \cite{gromov1987essays}, 
it is possible to extend such a  measure of curvature for manifolds to discrete networks.   
Hyperbolicity measure \cite{gromov1987essays} defines the curvature for an infinite metric space with bounded 
local geometry, using a 4 points condition.
In detail, the hyperbolicity $\delta(x,y,v,w)$ of a $4$-tuple of vertices $\{x,y,v,w\}$ is defined as half the difference between the biggest two of the following sums:
\begin{equation} \label{eq:sums}
d(x,y)+d(v,w),\ \ \  d(x,v)+d(y,w), \ \ \ d(x,w)+d(y,v)
\end{equation}
where $d$ denotes the distance between two vertices. If the lengths of edges are integers (in our case, the length of each edge is assumed to be 1), the hyperbolicity of a $4$-tuple of vertices is always an integer or its half.
The hyperbolicity $\delta(G)$ of a graph $G$ is commonly defined as the maximum of the hyperbolicity of a $4$-tuple of vertices \cite{bermudo2011computing, bermudo2011hyperbolicity, Carballosa2013, cohen2013exact, kennedy2013hyperbolicity}. However, for the purposes of this work, also the average hyperbolicity of a $4$-tuple will have a significant role: to distinguish the two, we will use $\dw$ and $\da$, respectively \cite{albert2014topological}.

This approach attracted the interest of the community, both in modeling this phenomenon \cite{shang2012lack}, and in classifying networks from the real world \cite{abu-ata2015metric}. For example, it has been argued that several properties of complex networks arise naturally, once a negative curvature of the space has been assumed \cite{krioukov2010hyperbolic}. Similarly, others investigated the role of hyperbolicity in a series of different networks \cite{albert2014topological} ranging from social networks in dolphins, to characters in books, with the aim of discovering essential edges in path of communication. In addition, by studying structural holes, it has been shown that most of these networks are essentially tree-like \cite{albert2014topological}. 

Nevertheless, nobody provided an interpretation of what is measured by hyperbolicity, and in the literature there are few applications of this quantity. 
This paper wants to fill this gap, by considering {\em hyperbolicity as a measure of how much a network is ``democratic''} (the larger the hyperbolicity value the larger the ``democracy'' in the network). 
Indeed, we prove mathematical results that link a small hyperbolicity constant with the existence of a small set of vertices ``controlling'' many paths, and hence with a non-democratic network (implications are true in both directions). 

As far as we know, this is the first measure of democracy in a complex network, apart from assortativity \cite{newman2002assortative,mixing2003newman}. In any case, our measure is quite different from the latter one, because it is based on shortest paths and not on neighbors: consequently, the new measure is global. Moreover, it is more robust: for instance, if we ``break'' all edges by ``adding a vertex in the middle'', the hyperbolicity of the graph does not change much, but the assortativity decreases drastically.

Starting from this interpretation, we derive consequences on the structure of biological, social, and technological networks, by computing $\dw$ and $\da$ on a dataset made by \datasetsize\ complex networks. 
This analysis confirms previous results showing that $\dw$ is highly influenced by ``random events'', and it does not capture a specific characteristic of the network \cite{albert2014topological}. Differently from previous works, we will also be able to quantify this phenomenon, showing that the distribution of $\frac{2\dw}{D}$ (where $D$ is the diameter of the graph) is approximately normal. The value of $\da$ is instead much more robust with respect to random events, and it allows us to effectively distinguish networks of different kind. Our classification will be different from the classification provided by assortativity \cite{newman2002assortative,mixing2003newman}: for instance, a network with few influential hubs not connected to each other is democratic if we consider assortativity, while it is aristocratic in our framework.

Finally, we introduce the hyperbolicity of vertex neighbors. This is done both starting from high-degree vertices and from random vertices. Our goal is to describe the ``influence area'' of such vertices. If a neighbor is not democratic, it means that this neighbor is in the influence area of the center (this implication is particularly evident when the center is a high-degree vertex).
If the starting vertex $v$ is random, the hyperbolicity of a neighbor grows quite irregularly, with local maxima when ``more influential'' vertices are reached, and local minima when the neighbor size corresponds to the influence area of such vertices. On the other hand, if the starting vertex is already influential, the plot grows almost linearly until the hyperbolicity of the whole graph is reached, and then it is constant. We have chosen to define the influence area of this vertex as the first neighbor where half this threshold is reached.

The size of the influence area of the vertex with highest degree provides a good way to classify networks. We are able to distinguish ``local'' networks, where each node creates connections to other nodes in order to reach their own goals, and ``global'' networks, where a common goal drives the creation of the network. For instance, peer-to-peer networks are local, because each node creates connections in order to download data, while metabolic networks are global, because the creation of links is driven by the global goal of making the cell alive. We show that influence areas in a local networks is quite small, containing about $\frac{n}{10}$ vertices, where $n$ is the total number of vertices. Conversely, influence areas in global networks have a much bigger size, close to $\frac{n}{3}$ vertices.




\section{Methods}

The main result of this paper is to show the link between the hyperbolicity of a graph and its  ``democracy''. This interpretation is motivated by the following propositions: the first one shows that, if for some vertices $v,w$, $\max_{x,y} \delta(x,y,v,w)$ is not high, then there is a set of small diameter that ``controls'' all approximately shortest paths from $x$ to $y$. Consequently, a network with low hyperbolicity is not ``democratic'', because shortest paths are controlled by small sets.
\begin{proposition}
Let $v,w$ be two vertices in a network $G=(V,E)$, let $B_r(v)$ be the $r$-neighborhood of $v$ (that is, the set $\{u \in V:d(u,v) \leq r\}$), and let $B_s(w)$ be the $s$-neighborhood of $w$. Then, the diameter of the set $X=B_r(v) \cap B_s(w)$ is at most $2\max_{x,y} \delta(x,y,v,w)+r+s-d(v,w)$.
\end{proposition}
\begin{proof}
Let $x,y$ be two vertices in $X$. Then, 
\begin{align*}
2\delta(x,y,v,w) \geq & d(x,y)+d(v,w)-\max(d(x,v)+d(y,w), \\
& d(x,w)+d(y,v)) \geq \\
\geq & d(x,y)+d(v,w)-(r+s).
\end{align*}

Taking the maximum over all possible $x,y \in X$, we obtain $D(X)=\max_{x,y \in X} d(x,y) \leq 2\max_{x,y}\delta(x,y,v,w)+r+s-d(v,w)$, where $D(X)$ is the diameter of the set $X$.
\end{proof}

The second proposition is a sort of converse: if there is a set of vertices controlling the shortest paths of a given $4$-tuple, the hyperbolicity of that $4$-tuple is low. Consequently, if the hyperbolicity is high, then there is not a small set of vertices controlling many shortest paths, and the network is democratic.

\begin{proposition}
Let $x,y,v,w$ be $4$-tuple of vertices, and let us assume that there exists a set $C\subseteq V$ of diameter $D$ such that all shortest paths between $x,y,v,w$ pass through $C$. Then, $\delta(x,y,v,w) \leq D$.
\end{proposition}
\begin{proof}
We may assume without loss of generality that $d(x,y)+d(v,w) \geq d(x,v)+d(y,w) \geq d(x,w)+d(y,v)$. Then, if we denote by $d(x,C):=\min_{c \in C} d(x,c)$,
\begin{align*}
2\delta(x,y,v,w) &= && d(x,y)+d(v,w) - d(x,v)-d(y,w) \\
& \leq && d(x,C)+D+d(C,y)+d(v,C)+D \\
&&& +d(C,w)-d(x,C)-d(C,v) \\
&&& -d(y,C)-d(C,w) \\
&=&&  2D.
\end{align*}
\end{proof}

These two results formalize our connection between the hyperbolicity constant and how democratic a complex network is. This work will confirm this interpretation by analyzing a dataset of \datasetsize\ graphs, made by \datasetbio\ biological networks, \datasetsoc\ social networks, and \datasettec\ technological networks, and it will draw conclusions on which networks are more democratic than others.


The first data check extends some of the activity already done \cite{krioukov2010hyperbolic,albert2014topological}: for each network in our dataset, we  computed the distribution of $\frac{2\dw}{\diam}$, where $D$ is the diameter of the graph (this value is always between $0$ and $1$ \cite{Fang2011}). With respect to the previous papers, we got more data referring to larger networks, and we therefore deal with lower statistical errors. 
This is particularly important, since it is known that $\dw$ does not capture structural properties of a network. Furthermore, its behavior is not robust, 
in the sense that small modifications in the graph can significantly change the value of $\dw$ \cite{albert2014topological}. Moreover, computing $\dw$ is difficult task, since the best ``practical'' algorithm \cite{cohen2013exact} has running time $\O(n^4)$. We have been able to overcome these issues by focusing on $\da$, already considered in the literature \cite{albert2014topological}, but not deeply analyzed.
In particular, for each graph in the dataset, we have considered the ratio $\frac{2\da}{\dista}$, where $\dista$ is the average distance between two randomly chosen vertices \cite{Fang2011} (this parameter takes values in the interval $[0,1]$).
Although the exact computation of $\da$ is also hard, the value $\frac{2\da}{\dista}$ can be easily approximated through sampling. More specifically, if we consider $N$ $4$-tuples of vertices with hyperbolicity $\delta_1,\dots,\delta_N$, and since $a_i=0 \leq \delta_i \leq 1=b_i$ for each $i$, by the Azuma-Hoeffding inequality \cite{hoeffding1963probability} we obtain:
\begin{align*}
\P\left(\left|\frac{\sum_{i=1}^N \delta_i}{N}-\da\right| \geq t \right) &\leq 2e^{\frac{-N^2t^2}{2\sum_{i=1}^N (b_i-a_i)^2}} \leq \\
& \leq 2e^{-\frac{Nt^2}{2}}.
\end{align*}
We choose to sample $N=10,000,000$ $4$-tuples of vertices, obtaining an estimate $\bar{\delta}_{\text{avg}}$. The previous inequality applied with $t=0.001$ yields:

\begin{align*}
\P\left(\left|\bar{\delta}_{\text{avg}}-\da\right| \geq 0.001 \right) &\leq 2e^{-5} \leq 0.7\%.
\end{align*}


This precision is more than sufficient for our purposes. 

Finally, we have analyzed neighbors of a given vertex $v$, approximating their ratio $\frac{2\da}{\dista}$ as before, in order to analyze the influence area of $v$. The size of the neighbor considered ranges from the degree of $v$ to $n$, with steps of $10$ vertices (where $n$ is the total number of vertices in the graph). We have chosen $v$ as the maximum degree vertex, which intuitively should have a big influence area, and we have used as a benchmark of comparison the same results from a random vertex.

We have defined the ``influence area'' of $v$ as the biggest neighbor where $\frac{2\da}{\dista}$ is at most half than the same value in the whole graph. However, in order to avoid ``random deviations'' (especially, when the neighbor is small), in our experiments we have considered the fourth neighbor where this event has occurred. The purpose of this analysis is twofold: not only we define and compute the influence area of a vertex, but we also classify networks according to the size of this influence area.






\section{Results}

\subsection{Worst-Case Hyperbolicity.}

The most common approach is to consider the maximum hyperbolicity of a 4-tuple of vertices, that is, $\dw(G)$. Despite some attempts in proving that real-world networks usually have low hyperbolicity \cite{kennedy2013hyperbolicity}, it soon became clear that the small values obtained are consequences of small world networks \cite{albert2014topological}, since $0 \leq 2\dw \leq D$ (where $D$ is the diameter of the graph). In particular, in a dataset of small social and biological networks \cite{albert2014topological} there is no relation between $2\dw$ and $D$. The ratio between these values varies between $25\%$ and $89\%$.

In this paper, we make a more detailed analysis, working with larger networks, between one hundred and several thousand of vertices. 
The results obtained are shown in Figure \ref{fig:deltaworst}, which also contains specific results dealing only with social, biological and technological networks.

These results show that the distribution of the ratio $\frac{2\dw}{D}$ is approximately Gaussian, both in the whole dataset and in each single kind of network. The average ratio is $0.521$, and the standard deviation is $0.085$. Moreover, a Chi-square goodness of fit test applied to the previous data does not reject the hypothesis that the distribution is Gaussian with mean $0.5$ and variance $0.085$, with a very high confidence level \cite{ross2010statistics}. This result confirms that the hyperbolicity of real-world networks is not much ``smaller than expected'', result already obtained in the past \cite{albert2014topological}. However, we are able to perform a further step: the Gaussian probability distribution makes us think that $\dw$ is influenced by random events. Indeed it does not reflect particular characteristics of the network, since the same distribution arises from networks of different kinds. 

\begin{figure*}
\includegraphics[width=\textwidth]{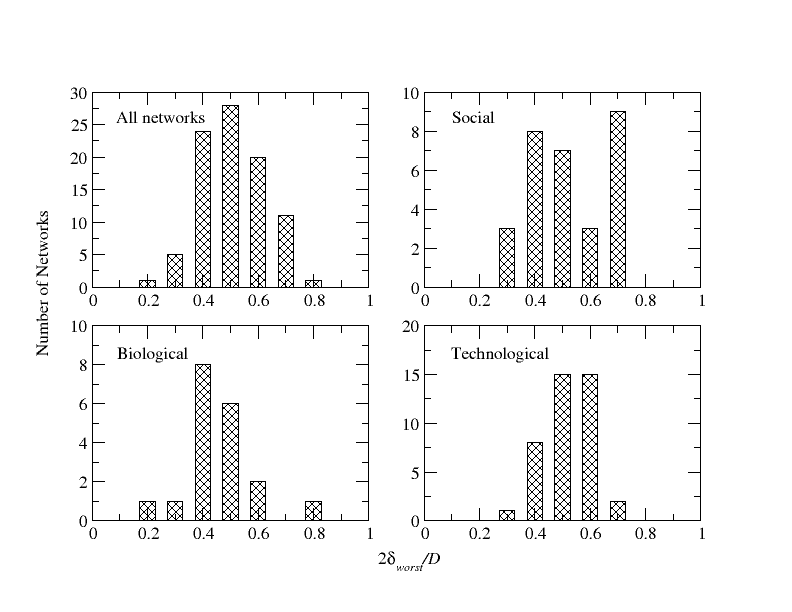}
\caption{The distribution of $\frac{2\dw}{D}$ in the graphs in our dataset. The bar corresponding to the value $p$ contains all networks where $p-0.5 < \frac{2\dw}{D} \leq p+0.5$.\label{fig:deltaworst}} 
\end{figure*}

Social networks show a slightly different behavior, since many of them have a larger value of $\frac{2\dw}{D}$, between $0.65$ and $0.75$. However, this is due to the presence of several financial (e-MID, a platform for interbank lending) networks, where the ratio is often $\frac{2}{3}$ or $\frac{3}{4}$ since the diameter is $3$ or $4$.

Despite this particular case, we may conclude that $\dw$ is not a characteristic of the network, but it mainly depends on ``random events'' that have a deep impact on the value of $\dw$.
This conclusion is further confirmed by the particular case of the e-MID networks: this parameter changed from $0.750$ in 2011 to $0.286$ in 2012, only because a simple path of length $3$ increased the diameter from $4$ to $7$.



\subsection{Average Hyperbolicity.}

In the past, the average hyperbolicity $\da$ of a 4-tuple of vertices was rarely analyzed: the only known result is that it is usually significantly smaller than $\dw$ \cite{albert2014topological}. 
However, we think that this parameter may provide very interesting results, because it is robust, in the sense that it does not change much if few edges of the graph are modified. Furthermore, it is easily approximable through sampling (while computing $\dw$ takes time $\O(n^4)$). 
Similarly to what we have done in the analysis of $\dw$, we have considered the ratio $\frac{2\da}{\dista}$, where $\dista$ denotes the average distance in the network (also this parameter lies in the interval $[0,1]$).
The results obtained are plotted in Figure \ref{fig:deltaavg}.

\begin{figure*}
\includegraphics[width=\textwidth]{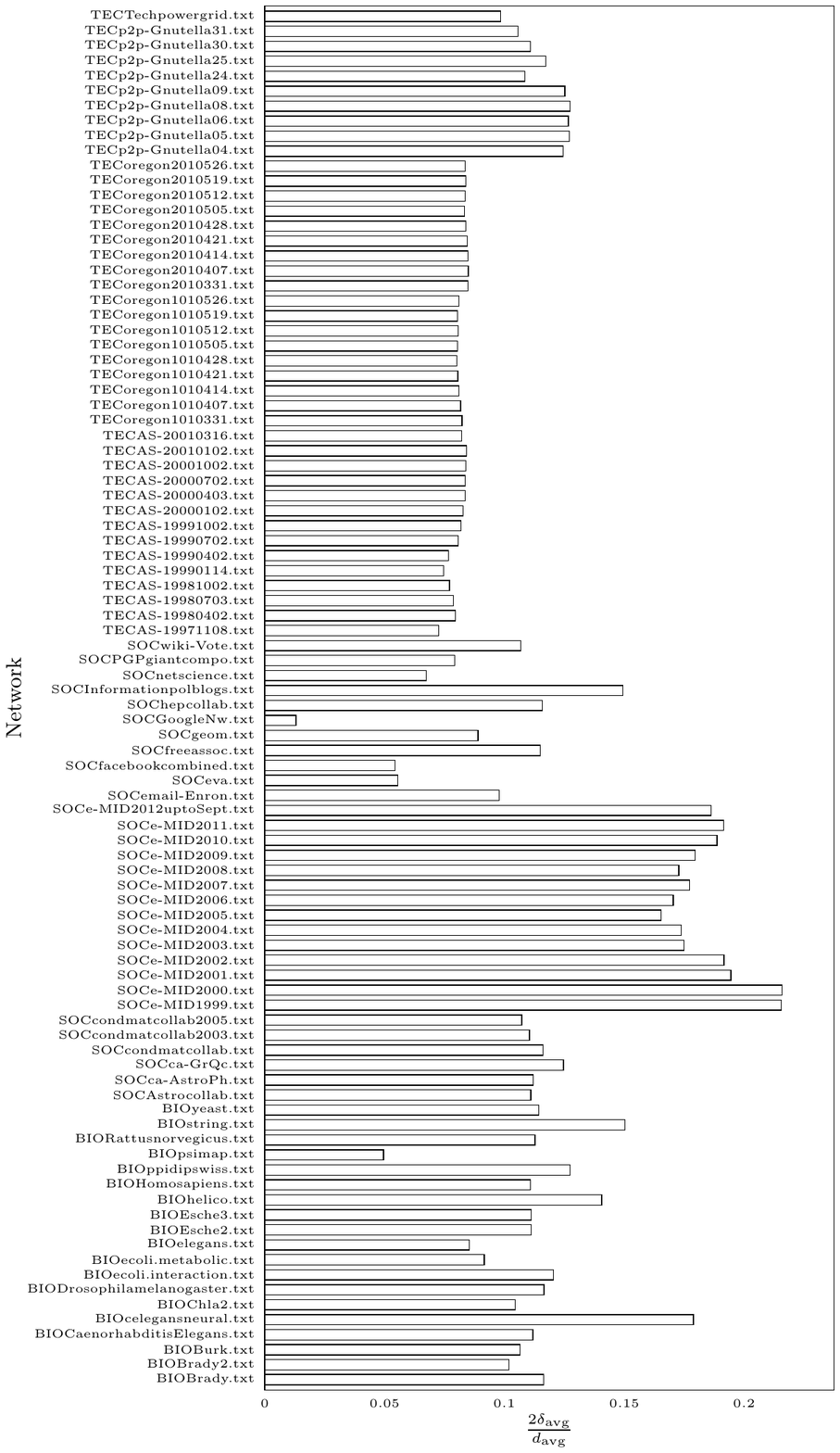}
\caption{The value $\frac{2\da}{\dista}$ of all the networks in our dataset.\label{fig:deltaavg}}
\end{figure*}
The picture shows that the average hyperbolicity is usually an order of magnitude smaller than the average distance: in this sense, real-world networks are indeed hyperbolic. Moreover, it is possible to make a distinction between networks that are ``more democratic'', like the e-MID networks or the peer-to-peer networks (where the hyperbolicity is large), and networks that are ``more centralized'', like some social networks and most autonomous systems networks.







\subsection{Hyperbolicity of Neighbors.}
Since the hyperbolicity of a graph is closely related to the existence of a small part of the graph controlling most shortest paths, we have analyzed which subgraphs of a given graph have small hyperbolicity. Intuitively, these subgraphs should be ``less democratic'' than the whole graph, in the sense that they are contained in the ``influence area'' of a small group of vertices. In this analysis, 
we have tried to spot the influence area of a single vertex, by measuring $\frac{2\da}{\dista}$ on neighbors of $v$ in increasing order of size.
In order to prove the effectiveness of this approach, we have first tested a synthetic power-law graph \cite{lancichinetti2009benchmarks} made by three communities of $1000$ vertices each (see the lowest plot in Figure \ref{fig:deltaneighsinglegraphs}). We have computed the hyperbolicity of neighbors of the vertex $v$ with highest degree: we can see a local minimum close to the size of a community. In our opinion, this minimum appears because the neighbor is ``dominated by the community'', and consequently by the center $v$ of the community. This result confirms the link between the value of $\frac{2\da}{\dista}$ and the influence area of a vertex.

Finally, we passed to the analysis of neighbors in real-world networks. The upper plots in Figure \ref{fig:deltaneighsinglegraphs} show the same results for one network of each kind: 
\begin{itemize}
\item a social network, the General Relativity and Quantum Cosmology collaboration network;
\item a biological network, the yeast metabolic network;
\item a technological network, the peer-to-peer Gnutella network in 2004.
\end{itemize}
As a benchmark of comparison, we have also considered the hyperbolicity of neighbors of a random vertex.

\begin{figure*}
\includegraphics[width=\textwidth]{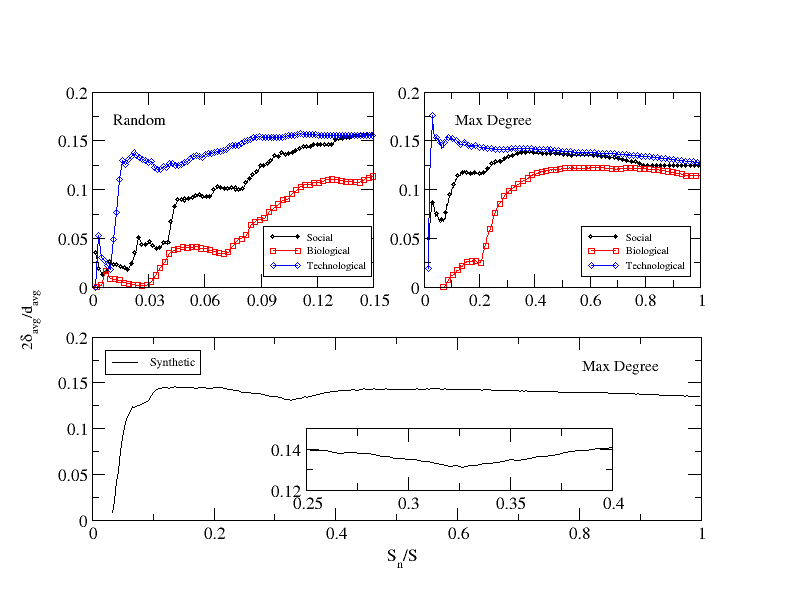}
\caption{The value of $\frac{2\da}{\dista}$ for neighbours of a randomly chosen vertex (up left), or the maximum degree vertex (up right); $S_n$ is the number of vertices in the neighbours, while $S$ is the total number of vertices. Results are shown for a social network, a biological network, a technological network, and (below) a synthetic network.\label{fig:deltaneighsinglegraphs}} 
\end{figure*}

The plots show that the hyperbolicity of a neighbor of the maximum degree vertex grows almost linearly with the neighbor size, until it converges to the hyperbolicity of the whole graph. Convergence time differs from graph to graph. In biological networks, convergence was reached at size close to $\frac{n}{2}$, while in the social and in the technological networks convergence is reached before. For neighbors of a random vertex, we outline a different behavior: at the beginning, the growth is not monotone, like in the previous case, and it is much more irregular. In our opinion, this is due to the fact that, when the neighbor grows, it reaches more and more ``influential'' vertices, and the first neighbor that touches such a vertex corresponds to a local maximum in the plot. After some steps, the hyperbolicity grows more and more regularly, because we have reached a very influential vertex $w$, and from that point on we are mainly considering the influence area of $w$, not of $v$. This issue is further confirmed by Figure \ref{fig:derivative}, where the derivative of the average hyperbolicity is shown. Hence, this experiment outlines two significant point: on the one hand, it allows us to define the influence area of $v$ as the neighbor where the first local maximum is reached; on the other hand, it gives motivations to the analysis of neighbors of influential vertices, because their influence area has a strong impact on the topology of the whole network.

\begin{figure}
\includegraphics[width=.48\textwidth]{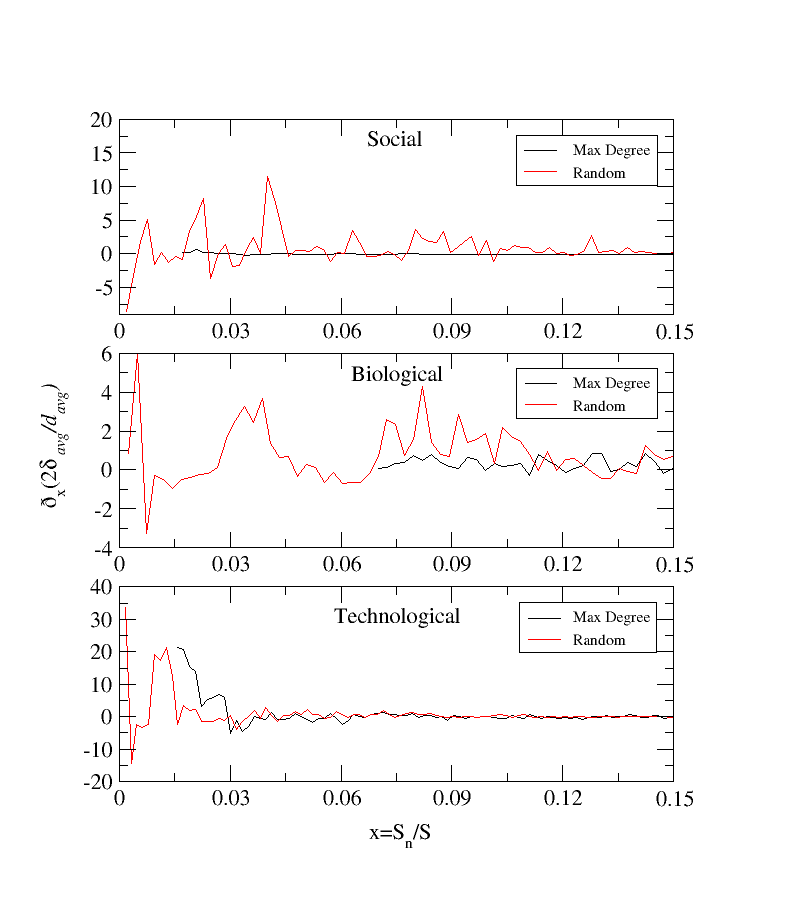}
\caption{The derivative with respect to the neighbors proportion $S_n/S$ of the value of $\frac{2\da}{\dista}$, in neighbors of the maximum degree vertex and of random vertices.
\label{fig:derivative}} 
\end{figure}


For this reason, we have focused on the maximum degree vertex, and, in order to have more general results, we have analyzed all graphs in the dataset. Figure \ref{fig:firstneigh} shows the size of the maximum neighbor having ratio $\frac{2\da}{\dista}$ at least half than the same ratio in the whole graph. Actually, in order to exclude random deviations from our analysis, we have plotted the fourth neighbor where the condition is satisfied.
The figure shows that the influence area of an individual is low in social and peer-to-peer networks, compared to biological or autonomous system network.



This standard behavior has few exceptions: first of all, protein-protein interaction networks (\texttt{string}, \texttt{ecoli.interaction}) are different from other biological networks, and the influence area is smaller. Furthermore, the social network \texttt{GoogleNW} contains a vertex with an enormous influence area: this network is the set of Google pages, and the central vertex $v$ considered is the page \texttt{www.google.com}, which clearly dominates all the others. Another particular case is the social network \texttt{facebook\_combined}: this network is a collection of ego-networks from Facebook, and links are made if common interests are retrieved. We think that this network is different from the others because it is a small subgraph of a bigger graph (where all Facebook users are considered), and the choice of the subgraph has a strong impact on the topology of the network, which does not reflect the standard behavior.

\begin{figure*}[h!]
\includegraphics[width=\textwidth]{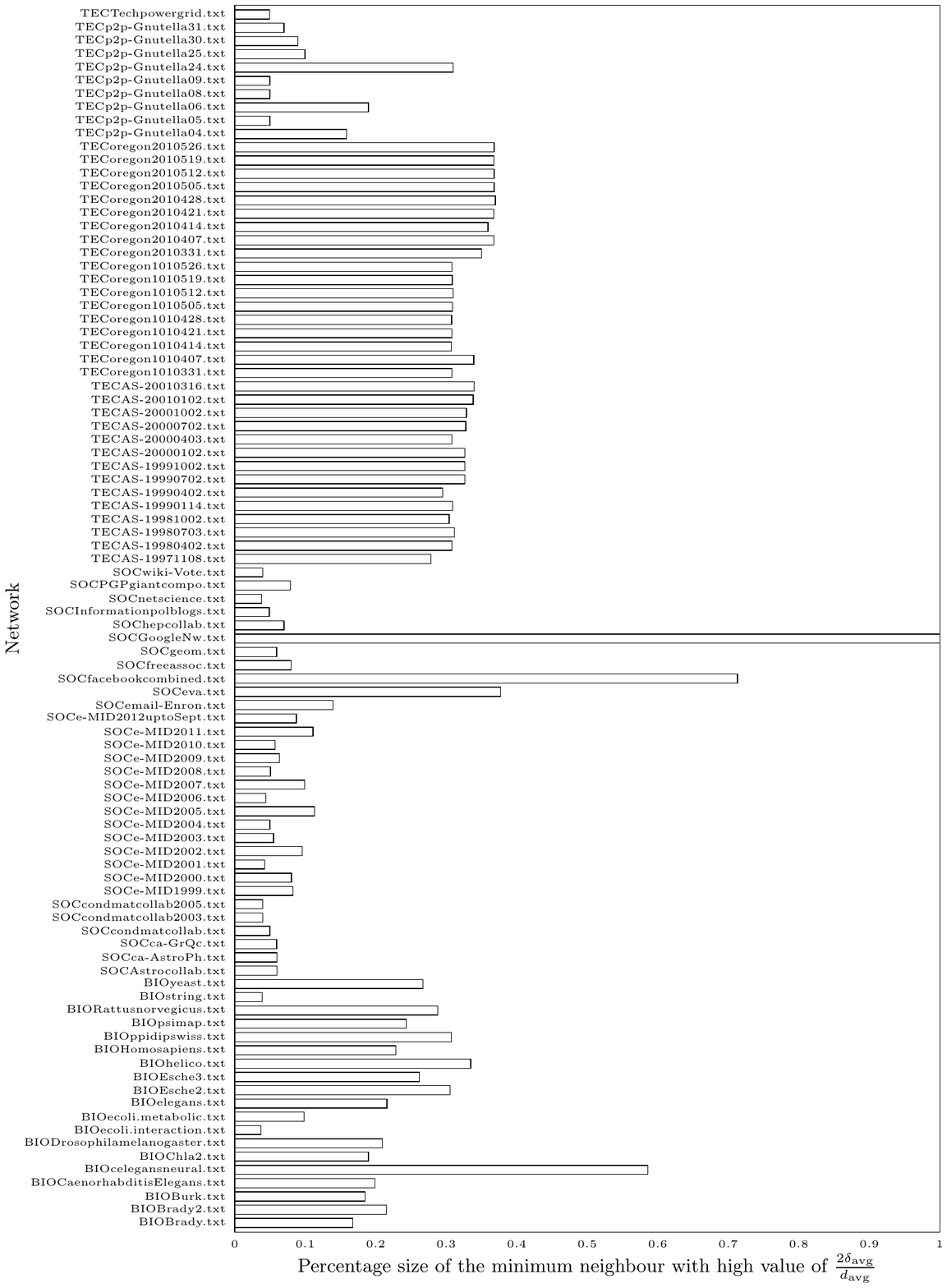}
\caption{The size of the fourth neighbour of the maximum degree vertex having $\frac{2\da}{\dista}$ at least half than the same value in the whole graph.\label{fig:firstneigh}}
\end{figure*}




\section{Discussion} \label{sec:IV}

In the literature, several works have analyzed the hyperbolicity of a complex network. They used this quantity in order to classify real-world networks and in order to draw conclusions about the impact of hyperbolicity on the network topology.
However, these works are mainly based on the analysis of $\dw$, which has two drawbacks: it is not \emph{robust}, that is, small modifications on the network can have deep impacts on its value, and it is not \emph{scalable}, that is, it can be exactly computed only on small networks.
In this work, we confirmed these conclusions, and we have proposed a different approach: using $\da$ instead of $\dw$, a parameter already considered in the literature. We interpreted this parameter as a measure of ``democracy'' in a network, and we classified different networks according to how democratic they are.

In particular, we have shown that technological autonomous system networks are less ``democratic'' than social or biological networks, in agreement with our intuition (since AS graphs have a ``built-in'' hierarchy, while in social networks everyone has the same role). Moreover, we have applied this concept to neighbors of influential nodes: we have seen that larger neighbors are more democratic, and the relationship between the size and the hyperbolicity is almost linear, until a threshold is reached. This analysis clearly outlines the influence area of a node, whose size strongly depends on the graph considered. We have shown that nodes have a rather small influence area in social and peer-to-peer networks, while in autonomous systems and biological networks the influence area can be close to half the graph. A possible explanation of this behavior is that the former networks are ``distributed'', in the sense that each node has a goal (downloading in peer-to-peer networks, and creating relationships in social networks), and edges are created locally by nodes that try to reach the goal. On the other hand, the latter networks have global goals (connecting everyone in the network, or making a cell live), and the creation of edges is ``centralized''. Our analysis is able to distinguish graphs of these two kinds.

These results prove that democracy in a complex network is well formalized by our definition of hyperbolicity, since the consequences of our interpretation are coherent with intuitive ideas.
Furthermore, we have provided an application of this interpretation, making it possible to define and analyze the influence area of a node. 
As far as we know, this is the first work that provides this interpretation of the average hyperbolicity of a graph.

Possible applications include not only the classification of networks according to this parameter, but also the classifications of nodes in a network, or the classification of different communities. These communities might be democratic, if everyone has ``the same role'' and the hyperbolicity is high, or not democratic, if there is a group of few nodes that keeps the community together, making the average hyperbolicity small.
Future researches may address this characterization, and they may also classify nodes in a network according to their influence area, in order to distinguish important hubs from peripheral node. Finally, it would be interesting to compare this approach with other approaches to determine the centrality of a node (closeness centrality, betweenness centrality, eigenvector centrality, and so on), in order to better outline the features of this analysis.

\bibliographystyle{abbrv}
\bibliography{library}

%


%
%
%
%
%
\end{document}